\documentclass{llncs}
\usepackage{color}
\usepackage[dvipsnames]{xcolor}
\usepackage{mathrsfs,amsmath}
\usepackage{amsfonts}
\usepackage{graphicx}
\usepackage[title]{appendix}
\usepackage[noend]{algpseudocode}
\usepackage{algorithm}
\usepackage{algpseudocode}

\makeatletter
\renewcommand{\ALG@name}{Pseudocode}
\makeatother

\newfloat{protocol}{htbp}{loa}
\floatname{protocol}{Protocol}
\makeatletter\newcommand\l@subroutine{\@dottedtocline{1}{1.5em}{2.3em}}\makeatother

\usepackage{subfigure}
\usepackage[caption=false]{subfig}

\usepackage{enumitem}
\usepackage{color}
\usepackage[dvipsnames]{xcolor}

\usepackage[switch]{lineno}

\usepackage{cancel}

\begin{document}

\title{MPC meets SNA: A Privacy Preserving Analysis of Distributed Sensitive Social Networks}
\titlerunning{Hamiltonian Mechanics}  
%
\author{Varsha Bhat Kukkala \and S.R.S Iyengar}

\institute{Indian Institute of Technology Ropar, India\\    \email{varsha.bhat@iitrpr.ac.in, sudarshan@iitrpr.ac.in}}
\authorrunning{Kukkala et al.} 
%
%

\maketitle              

\begin{abstract}
In this paper, we formalize the notion of distributed sensitive social networks (DSSNs), which encompasses networks like
enmity networks, financial transaction networks, supply chain
networks and sexual relationship networks. Compared to the well studied traditional social networks, DSSNs are often more
challenging to study, given the privacy concerns of the individuals
on whom the network is knit. In the current work, we envision
the use of secure multiparty tools and techniques for performing
privacy preserving social network analysis over DSSNs. As a step towards realizing this, we design efficient data-oblivious algorithms for computing the K-shell decomposition and the PageRank centrality measure for a given DSSN.
The designed data-oblivious algorithms can be translated into equivalent secure computation protocols. 
We also list a string of challenges that are needed to be addressed, for
employing secure computation protocols as a practical solution for studying DSSNs.
\end{abstract}
\section{Motivation}

Animosity between colleagues is inevitable and occurs in most organizations. 
These personal feelings are known to affect an individual\rq s professional decisions. 
For example, it is well studied that during the formation of teams, employees tend to choose likable colleagues as co-workers despite their incompetence, rather than picking a ``competent jerk'' \cite{casciaro2005competent}. 
In order to better understand employee dynamics, the animosity between them can be modeled as a social network, which we refer to as the \textit{enmity network}. 
The network comprises of a set of nodes and a set of edges, where each node is representative of an employee and an edge from node $A$ to $B$ depicts the hatred that employee $A$ has for employee $B$. 
The structure of the enmity network provides an overall view of how the employees relate to one another. Moreover, the importance of a study investigating the structure of informal networks on the employees, such as the enmity network, is well established \cite{krackhardt200116}. 
However, in the case of an enmity network, such a study is infeasible since the information regarding the network is both distributed and highly sensitive. 
The network data is distributed, in the sense that the presence or absence of an edge from $A$ to $B$ is known only to employee $A$. Thus, each employee is aware of only those edges in the network that emanate from the node corresponding to her. 
Additionally, the network is sensitive since employees refrain from disclosing their feelings of hatred and anger, as the professional setting in most organizations will consider this unacceptable. 
It is a must for the organization to address these feelings, because when ignored these could in turn manifest as passive-aggression, clearly undesirable for conducive work environment \cite{website2015}.  \\

Several  social networks, apart from enmity networks, pose similar challenges in their study. The common features across these networks, that make the study challenging, are that the network data is present in a distributed setting and that the data is  highly sensitive/private an information. We define this class of social networks as \textit{distributed sensitive social networks} (DSSNs).  
Supply chain networks of organizations, romantic relationships of individuals and the network of financial transactions between entities are examples of a few social networks that fall under this category. A detailed discussion on the notion of DSSNs is provided in Section \ref{dssn}. \\

A traditional approach adopted for analyzing distributed data is the \textit{trusted third party} (TTP) model. This model involves a central trusted authority which collects information from each individual who has a share of the complete data. This central authority, known as the trusted third party, aggregates the data and will either perform the required analysis or  release a sanitized version of the data for analysis.
This model has been successfully adopted in studying few of the social networks, such as online friendship networks, email communication networks, and collaboration networks. The presence or absence of an edge in a DSSN is associated with a high quotient of sensitivity. Hence, individuals holding the network data will be unwilling to disclose their private information to a trusted third party. The above claim is well supported by the results of a survey conducted as a part of the current work. A survey on the sensitivity of personal data was conducted to determine how private do individuals consider their data. The results of the survey are discussed in Section \ref{dssn}. This renders the TTP model infeasible as a solution for the analysis of DSSNs, necessitating for a new methodology to analyze social networks that are both distributed and sensitive in nature. The current work aims at proposing a new privacy preserving technique for analyzing distributed sensitive social networks. \\


	\section{Problem Statement}\label{problem_statement}

We model a \textit{distributed sensitive social network} as a \textit{graph} $G(V,E)$, where $V$ is the set of nodes and $E$ is the set of edges. Each node in $V$ is representative of an individual/organization on whom the network is considered. An edge $e \in E$ is an ordered pair of nodes $(u,v)$, which depicts the private interaction extending from node $u$ to node $v$. While modeling DSSNs, we assume that the set $V$ of vertices is publicly known and that the edges are directed\footnote{An undirected graph with bi-directional edges will fall under this definition as well}. 
The edges of the graph are considered private, given the sensitive nature of interactions that they model. 
Hence, the presence or absence of an edge $(u,v)$ could be private to node $u$ (as in the case of enmity network) or to both $u$ and $v$ (as in the case of financial transaction network). 
 The individuals who collectively retain the interaction information, i.e. the set $E$ of the graph $G(V,E)$, will henceforth be referred to as \textit{parties}. Given the distributed nature of a DSSN, we can model a party $P_i$ to possess as its private information a subgraph $G_i(V,E_i)$, where the edge set $E_i$ induces a partition $V_i$ of the vertex set and $E_i \subset E$ consists of all those edges $(u,v) \in E \ni' u \in V_i$. That is, $E_i$ consists of all those edges in $E$ that lie within the partition $V_i$ or those emanating from nodes in $V_i$. In the case of a DSSN where an edge $(u,v)$ is private to both $u$ and $v$, then the set $E_i$ will also include edges $(u,v) \ni' v \in V_i$, as in the case of a financial transaction network. $E_i$ in this case will additionally have the edges incoming to the partition $V_i$. 
The above definition gives a general picture of how the network data is distributedly held by the parties. The illustration in Figure \ref{fig:dist_fin_net} depicts a financial transaction network distributedly held by three parties (in this case banks). 
It is important to note that when parties are themselves individuals in the DSSN, the size of each partition is one i.e. for every $i$, $|V_i| = 1$.\\


\begin{figure}
\centering
\includegraphics[width=7cm,height=7cm]{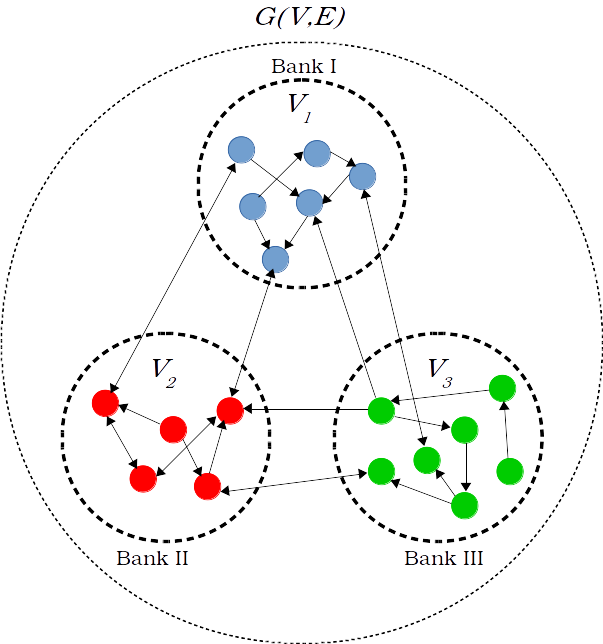}
\caption{The figure illustrates an example financial transaction network distributedly held by three parties - Bank I, Bank II and Bank III. The induced partitions are highlighted in the figure and each partition denotes the bank accounts belonging to a single bank. An edge $(u,v)$ corresponds to a financial transaction from account $u$ to $v$.}
\label{fig:dist_fin_net}
\end{figure}

A study of the structural characteristics of a graph representing a social network is defined as \textit{social network analysis} (SNA) \cite{scott2012social}. Most commonly employed SNA techniques include degree distribution, clustering, centrality measures and community detection \cite{scott2012social}.
Performing social network analysis on a distributed sensitive social network is akin to running a distributed algorithm $\mathcal{A}$, that takes as input the graph $G(V,E)$ representing the DSSN. 
The parties run the algorithm $\mathcal{A}$ distributedly by communicating and coordinating their actions with one another through passing of messages. However, unlike any regular distributed algorithm, $\mathcal{A}$ must ensure that certain security requirements are met. These requirements are enlisted below.\\

\begin{definition}{(Security Requirements)}\label{distributed_algo}
A distributed algorithm $\mathcal  {A}$ computes a network measure $m$ over an input DSSN $G(V,E)$ in a privacy-preserving manner if the following constraints are met:
\begin{enumerate}[label=(\alph*)]
\item\textbf{Privacy}: No party must learn anything more than its input and output to the algorithm. Intuitively, the exchange of messages between parties, throughout the run of $\mathcal{A}$, must not compromise any party $P_i$'s private information ($E_i$). 
\item \textbf{Correctness}: The output $m$ generated by algorithm $\mathcal{A}$ must indeed be the desired network measure, globally fixed by the parties prior to the run of the algorithm. 
\end{enumerate}
\end{definition}

The current work envisions a thorough investigation of the class of under-explored social networks that we classify as distributed sensitive social networks. We question the possibility of designing efficient distributed algorithms, described in Definition \ref{distributed_algo}, as a privacy preserving approach for performing social network analysis on DSSNs.


\section{Proposed Vision: MPC meets SNA}\label{vision}

The question being addressed falls under the broader umbrella of the field of secure multiparty computation (MPC). Secure computation deals with designing algorithms/protocols so that a set of $n$ parties (or individuals) $P_1, P_2, \dots P_n$ having private data $x_1,x_2,...,x_n$ respectively, can securely compute a public function $f(x_1,...,x_n)$. In our case, the public function $f$ is specific to network analysis algorithms and the private data is in the form of network interactions. One can design multiparty computation protocols for different security requirements. However, the minimum notion of security guaranteed by the protocols is that of privacy and correctness, as specified under Definition \ref{distributed_algo}. Besides, secure computation also accounts for the scenario of collusion involving collaboration between parties with the ill-intent to learn other parties' private data and the possibility of parties deviating from the actual protocol to disrupt the process. This allows one to design protocols while accounting for the level of security demanded by the application at hand. For a discussion on the various security models studied in MPC, we point the reader to the work by Lindell and Pinkas \cite{lindell2009secure}.\\

Generic MPC constructions for various security definitions exist to compute any arbitrary function $f$ securely \cite{goldreich1987play,chaum1988multiparty}. However, these constructions are generally not efficient enough to be put to use in practice. A more sought after approach is to design function specific efficient MPC protocols. Few application scenarios that have adopted this approach include auctions \cite{bogetoft2006practical}, benchmarking \cite{atallah2004private} and voting \cite{cramer1997secure}. As the vision of the current paper, we propose the following:
\vspace*{0.05cm}
\newline{\hspace*{0.2cm}
``\textit{Design of efficient multiparty computation protocols for performing a privacy preserving  social network analysis of distributed sensitive social networks. }''\vspace*{0.2cm}} 
\newline{Designing} efficient and secure MPC protocols for analyzing DSSNs requires a thorough understanding of the state of the art MPC tools and techniques, the access patterns of network algorithms and a list of requirements/constraints imposed by the application scenario. Hence, the current work envisions the amalgamation of ideas from two seemingly different domains of research, namely, secure computation (MPC) and social network analysis (SNA). In the current work, we also enumerate over the challenges that need to be overcome to use MPC as a practical solution to study DSSNs.  \\

\subsection*{Our Contribution}
In the current paper, we introduce the notion of DSSNs and highlight the need for studying these less explored networks. We provides a detailed discussion on the characterization of DSSNs (Section \ref{dssn}). As a step towards realizing the proposed vision, we design efficient data-oblivious algorithms for two of the widely employed network analysis methods, namely, the K-shell decomposition algorithm and the Pagerank centrality measure (Section \ref{algorithms}). We discuss on the translation of these oblivious algorithms into their equivalent MPC protocols (Section \ref{prelim}). The design of such MPC protocols is just the first step. We list a string of challenges to be addressed in order to fully realize the proposed vision ( Section \ref{challenges}). These challenges are not limited to the MPC setting, but address the different aspects of the general question of performing privacy preserving analysis of DSSNs. For completeness, we also enumerate over the marginal research that has been conducted previously in the intersection of network analysis and MPC in Section \ref{related_work}. \\


\section{Distributed Sensitive Social Networks}\label{dssn}

The class of distributed sensitive social networks are characterized by two properties, namely, the interactions in the network being sensitive  and the network data being distributedly held by a set of parties. Furthermore, the degree of sensitivity or distributedness varies across DSSNs, a brief discussion on which is presented next. \\

\paragraph*{Degree of sensitivity} This parameter represents the extent to which the parties holding the DSSN consider their data private. To capture the gray scale value of sensitivity of a DSSN modeled as $G(V,E)$, we define a coefficient of sensitivity $\mu(G)$ as the fraction of parties distributedly holding $G$ who are unwilling to share their private network information with a trusted third party. By definition, the coefficient $\mu(G)$ lies between 0 and 1 and it signifies  the extent to which using the TTP model would fail while studying the graph $G$. To calculate the sensitivity coefficient for various DSSNs, we conducted a survey over 160 participants, with $90\%$ of the participants being from the age group 17-22. The targeted age group is as specified before since it is observed that romantic relationships, crushes, enmity, etc. are commonly observed interactions across these age groups. The survey participants were asked whether they will share their private information (like email transactions, list of close friends, people they dislike and people they have a crush on) with an external agency, who is collecting this information for research purposes and guarantees to protect their privacy. A plot containing sensitivity coefficient of various DSSNs is available in Figure \ref{fig:survey_plot}. From the plot we can infer that while friendship, trust and email communication networks fare lower on the scale of sensitivity, DSSNs capturing enmity, romantic and sexual relationships are highly sensitive. This also explains why social networks like friendship and communication have been able to make a digital footprint and hence are well studied, while many other sensitive social networks continue to remain in the distributed setting.\\ 

\begin{figure}
\includegraphics[width=8cm]{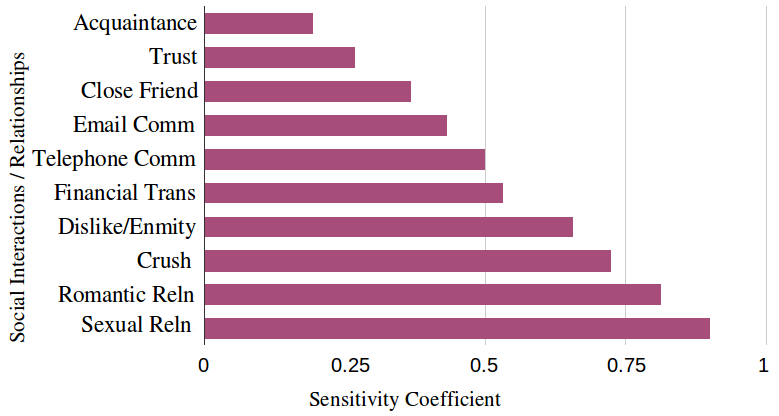}
\centering	
\caption{The plot depicts the coefficient of sensitivity for various DSSNs. Abbreviations used in the plot: Comm: Communication, Trans: Transaction, Reln: Relationship.}
\label{fig:survey_plot}
\end{figure}

\paragraph*{Degree of distributedness} The extent to which a network is distributedly held may vary across the networks. A sensitive social network is termed \textit{completely distributed} if each party represents a node in the social network. Examples include the trust network, romantic relationship network and the supply chain network. At times, the DSSN may be distributedly held by a small set of parties, such that each party has the network information over a subset of the nodes. These networks are termed \textit{partially distributed sensitive social networks}. Examples include the financial transaction network and the telephone communication network, where banks and telecommunication companies respectively are the associated parties. \\

A discussion on all the DSSNs that have appeared in the research literature is available in Appendix \ref{appendix_dssn}.


\section{Preliminaries}\label{prelim}
The paper proposes the design of secure multiparty computation protocols for performing privacy preserving network analysis. As a step towards realizing the proposed vision, we design data-oblivious algorithms for two commonly employed network measures. The current section focuses on introducing the notion of \textit{data-oblivious algorithms} and further we discuss on how to translate these oblivious algorithms to  secure MPC protocols with comparable efficiency. \\

\subsection*{Data-Oblivious Algorithms}

An algorithm is termed \textit{data-oblivious} if its control flow and memory access pattern do not leak any information about the input other than its length. Efficient data-oblivious algorithms can be employed for outsourcing computation \cite{eppstein2010privacy} and designing secure multiparty protocols \cite{blanton2013data}. In this paper, we  design data-oblivious network analysis algorithms which reveal nothing but the input length, which in our case are the number of nodes and the number of edges in the input graph. Further in this section we present some primitive data types and programming constructs that will be assumed throughout the paper. \\

The data-oblivious algorithms designed in the paper assume both integer and rational data types. A variable named with the symbol `tilde' over it will represent a rational data value ( For example $\overset{\sim}{a}$) and a variable without a tilde will represent an integer. Each data value is classified as \textit{public} or \textit{private}, where a public data value is not considered to be sensitive information. Hence, the control flow and memory access pattern of a data-oblivious algorithm can depend on public values as well, along with the input length. In the designed SNA algorithms, public and private values are colored \textcolor{ForestGreen}{green}  and \textcolor{red}{red} respectively.  We also assume certain primitive operations over integers and rational numbers, including addition($+$), multiplication($*$), division($/$), comparison($>,<$) and equality check($==$). The efficiency of a data-oblivious algorithms is evaluated in terms of the number of primitive operations employed by it. \\

Performing array accesses (A[$\textcolor{red}{i}$]) over secret indexes ($\textcolor{red}{i}$) is a  useful primitive for designing data-oblivious algorithms. Efficient array accesses over secret indexes can be performed using the Oblivious RAM (ORAM) primitive \cite{stefanov2013path,wang2015circuit,liu2014automating}. Given an \textcolor{ForestGreen}{$n$} length array $A$ and the index $\textcolor{red}{i}$, the block $A[\textcolor{red}{i}]$ can be accessed data-obliviously using an ORAM in $poly.log(n)$ primitive operations. Henceforth, $f(n)$ denotes the number of primitive operations employed while reading/writing to an element in an ORAM of size $n$. In all the data-oblivious algorithms, an array stored in the ORAM is colored  \textcolor{Violet}{blue}. \\

The designed data-oblivious algorithms employ two programming constructs, namely, \textit{for} and \textit{if-else}. All the \textit{for} loops employed are over public variables and hence the \textit{for} loop  is nothing but a succinct representation of a set of statements that are repeatedly executed  a publicly known number of times. Converting the non-oblivious \textit{if-else} construct into its oblivious counterpart has been well studied in the literature \cite{liu2014automating}. Intuitively, the \textit{if-else} construct can be made oblivious by executing both the branches  \textit{if} and  \textit{else} one after the other, while ensuring that the effect of only one branch takes place and the other is executed in a dummy fashion. For the sake of brevity, we use the non-oblivious  \textit{if-else} construct in the designed algorithms, however, they are nothing but an alias for their oblivious variants. \\

\subsection*{Translating Data-Oblivious Algorithms to Equivalent MPC Protocols}

MPC protocols are generally designed using ideal functionalities, which provide certain secure primitive operations as building blocks. All the arithmetic operations  over integers can be securely computed using the extended arithmetic black box ($\mathcal{F}_{ABB}$) \cite{damgaard2006unconditionally,damgaard2003universally}. The analogous ideal functionality for rational numbers ($\overset{\sim}{\mathcal{F}}_{ABB}$) was introduced by Catrina and Saxena \cite{catrina2010secure}. The ideal functionality for ORAM array accesses ($\mathcal{F}_{ORAM}$) has been well explored in the research literature as well \cite{wang2015circuit,keller2015oblivious}. Hence, there exists equivalent primitive MPC building blocks for each primitive operation assumed in the data-oblivious algorithms ($+,*,/,<,>,==$). This provides a natural procedure to convert any data-oblivious algorithm to its equivalent MPC protocol. The security of the designed protocols will follow from the UC theorem \cite{canetti2001universally} and the fact that the sequence of primitive operations employed in the designed oblivious algorithm is predetermined and depends only on public values. Due to space constraint, we have refrained from elaborating on certain concepts including the security definition of MPC protocols, the notion of ideal functionalities and the UC theorem. However, a discussion on them can be found in the work by Damagard et al. \cite{mpcbook}. \\ 


\section{Design and Analysis of Data-Oblivious SNA Algorithms}\label{algorithms}

Computing the centrality measure of nodes is a popular network analysis technique for finding important nodes in a network. However, the notion of importance is dependent on the kind of interactions being modeled as well as the application scenario under study. Thus, there exist numerous centrality measures to rank nodes with respect to different application criteria. In the current section, we present data-oblivious algorithms for computing two of the commonly used centrality measures, K-shell decomposition and Google PageRank. \\

For each network measure of interest, we first present the most widely employed algorithm for computing the network measure and discuss on which steps in each algorithm are non-oblivious to the input graph. Further, we present its oblivious counterparts designed using primitive black-box operations, as mentioned in the Preliminary Section \ref{prelim}. Finally, a discussion on the correctness, obliviousness and efficiency for each designed protocol is provided.

\subsection*{Oblivious graph data representation}

 \textit{Adjacency list} and \textit{adjacency matrix} are two widely employed graph data representations. Since the length of an adjacency list leaks the degree of the corresponding vertex, the adjacency list data structure is unsuitable for the oblivious setting. While the adjacency matrix is reasonably oblivious (it discloses only the number of nodes), it is space inefficient.  The adjacency matrix is never used in practice given that most real world network are sparse i.e have $|E| = \mathcal{O}(|V|)$. In this section we propose the \textit{edgelist graph representation},  a new graph data structure that outperforms both the adjacency list and adjacency matrix representations in the oblivious setting.

\begin{figure}
\begin{tabular}{@{}cccc@{}}
\includegraphics[width=0.2\linewidth]{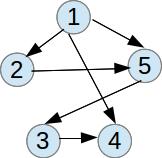}~~ &
\includegraphics[width=0.2\linewidth]{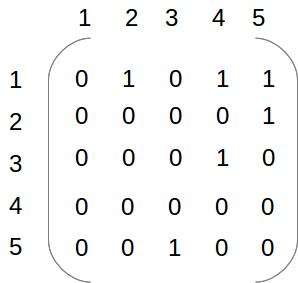} ~~&
\includegraphics[width=0.2\linewidth]{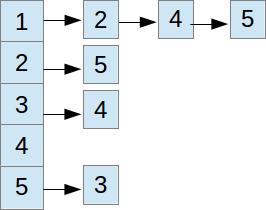} ~~&
\includegraphics[width=0.3\linewidth]
{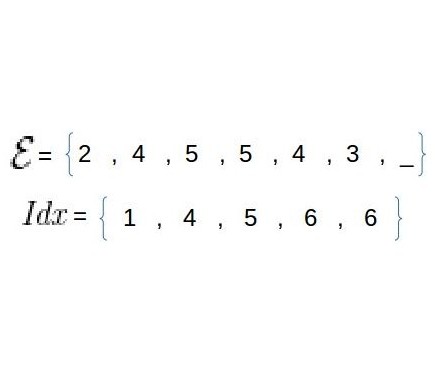}\\
(a) & (b) & (c) & (d)
\end{tabular}
\caption{The figures represent the different representations of a graph. }
\end{figure}


\paragraph*{Edgelist graph representation} The graph $G(V,E)$ (with $n = |V|$ and $m = |E|$) is represented using a two tuple $(\mathcal{E}, Idx)$, where $\mathcal{E}$ is an array of length $(m+1)$ containing concatenated adjacency lists of all nodes and an empty element at the end. The $Idx$ array, of length $n$, has an entry corresponding to each node as defined below:\\ 

$Idx[u] =
\left\{
	\begin{array}{ll}
		i  & \text{if out-degree($u$)$\neq 0$ and $E[i]$}\\ 
            & \text{is the first neighbor of $u$ in $E$ }\\
		Idx[u+1] & \text{$u \neq n$ and out-degree($u$)$ = 0$} \\
         m+1 & \text{$u = n$ and out-degree($u$)$ = 0$ }
	\end{array}
\right.$ \\

The edgelist graph representation is oblivious to the graph structure and reveals only the number of nodes and the number of edges in the network. Its space complexity is $(m+n+1)$, making it efficient for storing sparse graphs. All the designed network analysis protocols in this paper assume that the input graph is represented in the edgelist format and is stored in the ORAM. For the analogous MPC setting, we do not discuss on how the distributedly held network can be stored in the edgelist representation in the $\mathcal{F}_{ORAM}$, since the exact method adopted will depend on whether the network is completely or partially distributed. Kukkala et al. \cite{kukkala2016privacy} discuss on how to store the adjacency matrix of a distributedly held network  in the the $\mathcal{F}_{ABB}$. A similar approach can be adopted for	 storing the edgelist representation of a network that is distributedly held.

\subsection{K-Shell Decomposition Algorithm}

As an alternative to density as a measure for cohesiveness of the network, Seidman proposed the idea of  \textit{k-cores} that recursively decomposes the network into cohesive subgraphs known as cores \cite{seidman1983network}.
The basic idea is to recursively construct the subgraph of a graph such that the measure of cohesiveness increases as we recurse through. Thus, the process fragments the network into a sequence of subgraphs such that each is a subgraph of the preceding graph and has higher cohesiveness. Decomposing the graph into cores also helps in better identifying the most influential nodes in the network \cite{pei2013spreading}. The innermost core nodes are known to be influential spreaders of information, diseases, etc. A formal definition of k-cores and the algorithm to generate them are discussed further.\\

\begin{definition}
Given a graph $G(V,E)$, consider the subgraph $G_k$ induced by the largest vertex set $V_k$, such that $V_k \subseteq V$ and $\forall v \in V_k$ the degree of $v$ in $G_k$ is at least $k$. Such an induced subgraph $G_k$ is called the \textit{k-core} of graph $G$.
\end{definition}

\noindent From the above definition, it is clear that we can construct a sequence of subgraphs as shown below, where $G_i \leq G_j$ denotes that the graph $G_i$ is a subgraph of graph $G_j$:
\[ G_k \leq G_{k-1} \leq \dots G_2 \leq G_1 \leq G_0=G(V,E) \]

\begin{definition}\label{shell_set}
Given $G_i$ and $G_{i+1}$ are the $i$ and $i+1$ cores of a graph respectively, then :
\[ shell_i=\{v \in G_i | v \in G_{i} - G_{i+1} \} \]
\end{definition}
\noindent Thus, $shell_i$ denotes  the set of all nodes in the graph that belong to $i$-core but not the $(i+1)$-core. This process of decomposing the graph into cores assigns each node a unique shell number.

\begin{definition}
Given a graph $G(V,E)$, For every vertex  $ v \in V$ we define $shell(v)=i$ if $v \in shell_i$, where $shell_i$ is as specified in Definition \ref{shell_set}.\\ 
\end{definition}

Batagelj and Zaversnik were the first to propose an efficient technique known as the K-shell decomposition algorithm \cite{batagelj2003m} that determines the cores of a given network. 
The idea of the algorithm is to prune nodes with degree lesser than a threshold, in each step. The nodes pruned at step $i$ are assigned to $shell_i$ and the graph that remains constitutes the $(i+1)$-core. The process repeats by considering the $(i+1)$-core. The details of the algorithm are given in the above pseudocode. \\

\begin{algorithm}
\caption{K-shell decomposition}
\begin{algorithmic}[1]
\Require{Graph $G(V,E)$ with $|V|=n$ and $|E|=m$}
\Ensure{ $shell(v)$, for each $v \in V$  }	
\State $V \leftarrow$ sorted list of vertices based on their degree
\For  { each $v \in V$ in the order} 
\State $shell(v) \leftarrow degree (v)$
\For { $u \in Adj(v)$}  \Comment{for all neighbors of vertex $v$}
\If {degree($u$) $>$ degree($v$)}
\State $degree(u) \leftarrow degree(u) - 1 $
\EndIf
\EndFor
\State $G \leftarrow G - \{v\}$  \Comment{Prune vertex $v$}
\State $V \leftarrow$ sort vertices based on their degree
\EndFor
\end{algorithmic}
\label{pseudo_kshel}
\end{algorithm}

We observe that the pseudocode for K-shell has two nested $for$ loops, in steps 2 and 4. The outer one loops over the list of vertices, sorted with respect to degree and the inner $for$ loops over neighbors of the current vertex. Hence, the number of times the inner $for$ loop is executed is dependent on the node currently set to be pruned in the outer loop. This shows that the algorithm is input graph dependent, making it non-oblivious. 
In order to make the algorithm data oblivious, we use the technique of \textit{loop coalescing} \cite{liu2015oblivm}, which involves combining both the loops while guaranteeing the correct functionality of the algorithm. We provide the oblivious equivalent in Protocol \ref{alg:oblivious_kshell}. The idea is to coalesce the loops such that, in each iteration of the coalesced loop, we will either access an element of the sorted vertex list or access a neighbor of the current node through the edge list $\mathcal{E}$. To decide which access needs to be performed in the current iteration, we keep track of two indexes $i$ for indexing elements of sorted list of vertices $vert$ and $j$ to index over edge list $\mathcal{E}$.

\begin{protocol}\label{kshell}
\caption{K-shell decomposition}
\begin{algorithmic}[1]
\Require{ Graph ($\textcolor{Violet}{\mathcal{E}}$, $\textcolor{Violet}{Idx}$)}
\Ensure{The shell number of each of the nodes $v \in V$ denoted by $shell(v)$ }	
\State $\textcolor{Violet}{deg}[\textcolor{ForestGreen}{i}] \leftarrow \textcolor{ForestGreen}{0}$;    $\hspace{5mm}\forall\  \textcolor{ForestGreen}{1} \leq \textcolor{ForestGreen}{i} \leq \textcolor{ForestGreen}{n}$
\State $\textcolor{Violet}{vert}[\textcolor{ForestGreen}{i}] \leftarrow  \textcolor{ForestGreen}{i}$;    $\hspace{5mm}\forall\  \textcolor{ForestGreen}{1} \leq \textcolor{ForestGreen}{i} \leq \textcolor{ForestGreen}{n}$
\State $\textcolor{Violet}{bin}[\textcolor{ForestGreen}{i}] \leftarrow 0$;  $\hspace{5mm} \forall\  \textcolor{ForestGreen}{1} \leq \textcolor{ForestGreen}{i} \leq \textcolor{ForestGreen}{n}-\textcolor{ForestGreen}{1}$
\For { $\textcolor{ForestGreen}{i}=\textcolor{ForestGreen}{1}$ to $\textcolor{ForestGreen}{n}$} 
\State $\textcolor{Violet}{deg}[\textcolor{ForestGreen}{i}]  \leftarrow  \textcolor{Violet}{Idx}[\textcolor{ForestGreen}{i}+\textcolor{ForestGreen}{1}] - \textcolor{Violet}{Idx}[\textcolor{ForestGreen}{i}]$
\State $\textcolor{Violet}{bin}[\textcolor{Violet}{deg}[\textcolor{ForestGreen}{i}]] \leftarrow \textcolor{Violet}{bin}[\textcolor{Violet}{deg}[\textcolor{ForestGreen}{i}]] + \textcolor{ForestGreen}{1}$
\EndFor
\State $\textcolor{red}{start}  \leftarrow  1$
\For {$\textcolor{ForestGreen}{i} \leftarrow \textcolor{ForestGreen}{1}$ to $\textcolor{ForestGreen}{n}$}
\State $\textcolor{red}{temp}  \leftarrow  \textcolor{Violet}{bin}[\textcolor{ForestGreen}{i}]$
\State $\textcolor{Violet}{bin}[\textcolor{ForestGreen}{i}]  \leftarrow  \textcolor{red}{start}$
\State $\textcolor{red}{start}  \leftarrow  \textcolor{red}{start} + \textcolor{red}{temp}$
\EndFor
\State $\textcolor{Violet}{vert} \leftarrow $ $oblivious\_sort$ $(\textcolor{Violet}{vert}, \textcolor{Violet}{deg})$
\State $\textcolor{Violet}{pos}[\textcolor{Violet}{vert}[\textcolor{ForestGreen}{i}]] \leftarrow \textcolor{ForestGreen}{i}$, $\forall \ \textcolor{ForestGreen}{1} \leq \textcolor{ForestGreen}{i} \leq \textcolor{ForestGreen}{n}$
\State $\textcolor{ForestGreen}{i}  \leftarrow  \textcolor{ForestGreen}{1}$;    $\textcolor{red}{v} \leftarrow \textcolor{Violet}{vert}[\textcolor{ForestGreen}{i}]$;     $\textcolor{red}{j}\leftarrow \textcolor{Violet}{Idx}[\textcolor{red}{v}]$
\For {$\textcolor{ForestGreen}{iter}  \leftarrow  1$ to $\textcolor{ForestGreen}{n}+\textcolor{ForestGreen}{m}-\textcolor{ForestGreen}{1}$}
\If {$(\textcolor{red}{j}== \textcolor{Violet}{Idx}(\textcolor{red}{v}+\textcolor{ForestGreen}{1}))$}
\State $\textcolor{ForestGreen}{i} \leftarrow \textcolor{ForestGreen}{i}+\textcolor{ForestGreen}{1}$
\State $\textcolor{red}{v}  \leftarrow  \textcolor{Violet}{vert}[\textcolor{ForestGreen}{i}]$; $\textcolor{red}{j}  \leftarrow  \textcolor{Violet}{Idx}[\textcolor{red}{v}]$
\Else 
\State $\textcolor{red}{u} \leftarrow \textcolor{Violet}{E}[\textcolor{red}{j}]$
\If {$(\textcolor{Violet}{deg}[\textcolor{red}{u}] > \textcolor{Violet}{deg}[\textcolor{red}{v}])$}
\State $\textcolor{red}{du}  \leftarrow  \textcolor{Violet}{deg}[\textcolor{red}{u}]$; $\textcolor{red}{pu}  \leftarrow  \textcolor{Violet}{pos}[\textcolor{red}{u}]$
\State $\textcolor{red}{pw}  \leftarrow  \textcolor{Violet}{bin}[\textcolor{red}{du}]$; $\textcolor{red}{w} \leftarrow \textcolor{Violet}{vert}[\textcolor{red}{pw}]$
\If {($\textcolor{red}{u}\ \bcancel{==} \  \textcolor{red}{w}$)}
\State $\textcolor{Violet}{pos}[\textcolor{red}{u}]  \leftarrow  \textcolor{red}{pw}$; $\textcolor{Violet}{vert}[\textcolor{red}{pu}] \leftarrow \textcolor{red}{w}$
\State $\textcolor{Violet}{pos}[\textcolor{red}{w}]  \leftarrow  \textcolor{red}{pu}$; $\textcolor{Violet}{vert}[\textcolor{red}{pw}]  \leftarrow  \textcolor{red}{u}$
\EndIf
\State $\textcolor{Violet}{bin}[\textcolor{red}{du}] \leftarrow  \textcolor{Violet}{bin}[\textcolor{red}{du}]+ \textcolor{ForestGreen}{1}$
\State $\textcolor{Violet}{deg}[\textcolor{red}{u}] \leftarrow \textcolor{Violet}{deg}[\textcolor{red}{u}] - \textcolor{ForestGreen}{1}$
\EndIf
\State $\textcolor{red}{j} \leftarrow \textcolor{red}{j}+\textcolor{ForestGreen}{1}$
\EndIf
\EndFor
\end{algorithmic}
\label{alg:oblivious_kshell}
\end{protocol}

\begin{theorem}
The proposed Protocol \ref{alg:oblivious_kshell} correctly computes the shell number denoted by $shell(v)$~ $\forall \  v \in V$,  for the input graph $G(V,E)$.
\end{theorem}
\begin{proof}
In steps 1-3, we initialize a few variables - $deg$ list to maintain the degree of all nodes; $vert$ list to maintain vertices in the sorted order of their degrees; $bin$ list is intended to maintain a pointer to the first entry in $vert$ having degree $i, \forall 1\leq i \leq n-1$. In steps 4-6, we compute the degree of each node along with the number of nodes of a particular degree, which is stored in $bin$. Using this information, in steps 7-11 we compute the starting index of nodes of a particular degree, that terminates the initialization of $bin$ list. In step 12, we sort the list of vertices in $vert$ with respect to their degree. The assumed data-oblivious sorting primitive has been well studied in the literature \cite{goodrich2010randomized,leighton2014introduction}. In step 13, we initialize a list $pos$ such that $pos[i]$ will store the index of vertex $i$ in the sorted list $vert$. The variable $i$ will be used to denote the index to scan over the $vert$ list, and is initialized to the  first element of $vert$ while $j$ will be used to scan over the edge list $\mathcal{E}$. The variable $v$ denotes the vertex that is currently being pruned. The \textit{for} loop in step 15 is the coalesced loop equivalent of the nested loops in the pseudocode for K-shell. We iterate over the loop $(n+m-1)$ times, where the \textit{if-else} branch at step 16 determines if the current iteration scans over $vert$ or over the array $\mathcal{E}$. In steps 19-28, we process the neighbor of node $v$ stored in $u$. We reduce the degree of $u$ in case it is of a higher degree than $v$ and perform swap with elements in $vert$ to place $u$ in its correct position after reduction in degree. The functionality is the same as that achieved by the pseudocode of K-shell. 
\end{proof}

\begin{theorem}
The proposed Protocol \ref{alg:oblivious_kshell} for performing K-shell decomposition is data-oblivious.
\end{theorem}
\begin{proof}
All the arrays/lists - $deg$, $bin$, $vert$, $pos$,  $\mathcal{E}$, $Idx$ are stored in the ORAM. Hence, the access pattern of these arrays donot leak any information about the elements being accessed. Apart from accessing elements of the arrays in steps 1-5, we perform addition and subtraction operations that are assumed as primitive operations and hence are oblivious. In step 6, we sort the array $vert$, which can be performed data-obliviously \cite{goodrich2010randomized,leighton2014introduction}. 
In steps 7-10, the \textit{for} loop runs for a publicly known number of times, accesses the elements of the $bin$ array sequentially and uses the primitive addition operation. Hence the sequence of operations remain the same for a fixed sized input and therefore remain oblivious. Given that the number of nodes and edges in the graph are publicly known, the loop in step 14 runs for a fixed number of iterations. As described in Section \ref{prelim}, the \textit{if-else} construct can be assumed to be data-oblivious.  
\end{proof}

\begin{theorem}
The Protocol \ref{alg:oblivious_kshell} takes   $\mathcal{O} \left((n+m) \left(f(n)+f(m) \right) + n(\log (n))^2 \right)$ time complexity, where $f(n)$ denotes the overhead for reading/writing to an element of ORAM array of size $n$. 
\end{theorem}

\begin{proof}
In steps 1-5 and 7-10, we use $\mathcal{O}(n)$ ORAM array accesses over arrays of size $n$ and hence the complexity would be $\mathcal{O}(nf(n))$ primitive operations. In step 11, we perform oblivious sorting whose complexity is assumed to be $\mathcal{O}(n (\log (n))^2)$ \cite{leighton2014introduction}.  The loop in step 14 is executed $(n+m-1)$ times, where each iteration employs a constant number of comparison, equality check, addition operations and ORAM accesses. These instructions cost $\mathcal{O}\left((n+m)(f(n)+f(m)) \right)$ primitive operations. From this we conclude that the algorithm has a complexity of $\mathcal{O}\left((n+m)(f(n)+f(m))  + n(\log (n))^2)\right)$ in terms of the primitive operations used. 
\end{proof}

\begin{corollary}
The  Protocol 1  for computing the K-shell decomposition uses $\mathcal{O}( (n+m) \log^2 (n) )$ primitive operations assuming  the ORAM accesses are performed using the Circuit ORAM construction.
\end{corollary}

\begin{proof}
The proof of this corollary follows from the fact that using the Circuit ORAM each block access can be performed in $\mathcal{O}(\log^2n)$ operations, where $n$ is the size of the ORAM array.
\end{proof}


\subsection{PageRank Centrality}
Page et al. \cite{page1999pagerank} proposed the PageRank centrality measure to rank the web-pages using the underlying hyperlink structure. The World Wide Web can be visualized as a network, with web-pages being the nodes and the hyperlinks forming the edges. Intuitively, the PageRank defines the importance of a node to be proportional to the sum of the PageRank centralities of nodes with incoming links to the considered node. Since its inception, it has been used to find central nodes in networks including citation networks and social networks  \cite{ma2008bringing}. We adopt the \textit{iterative method} to compute PageRank, the pseudocode for which has been described below. The method is briefly described here. All the nodes begin with an equal share of PageRank value. In each iteration, every vertex distributes its PageRank value to all its neighbors (out-going links). Further, a fraction $(1-s)$ of the PageRank values of all the vertices is removed and redistributed to all the vertices equally. The uniform redistribution of the PageRank value of $(1-s)$ is performed to ensure that a zero out-degree vertex does not absorb all the PageRank values. As the number of iterations increase, the PageRank values will converge. 

Let $G(V,E)$ represent the graph of interest.The out-degree of vertex $u\in V$ is represented as $od(u)$.

\begin{definition}{(Update matrix $N$)}\label{update_matrix}
The update matrix for a graph $G$ is denoted by $N = [N_{ij}]$ with size $n \times n$, where $n = |V|$. The entries of the matrix are defined as follows: \\ 

$N_{ij} =
\left\{
	\begin{array}{ll}
		(s)/od(i)+(1-s)/n  & \mbox{  if  $\{i,j\} \in E$  }\\
		(1-s)/n & \mbox{  otherwise  } 
         
	\end{array}
\right.$ \\

\end{definition}

\begin{definition}{(PageRank update rule)}\label{update_rule}
Let $r_i^{(k)}$ represent the PageRank value of node $i$ after $k$ iterations. 
The PageRank update rule is defined as follows:
\begin{align*}
r_i^{(k)} \leftarrow \sum_{j=1}^{n} N_{ji} r_j^{(k-1)}
\end{align*}
\end{definition}

\begin{algorithm}[H]
\caption{PageRank algorithm (Adjacency matrix representation)}
\begin{algorithmic}[1]
\Require{Graph $G(V,E)$, $l$ (number of iterations), $s$ (redistribution parameter)}
\Ensure{The PageRank values for all the vertices in the graph $G$ }	
\State Compute the matrix $N = [N_{ij}]$ as specified in Definition \ref{update_matrix}
\State Intialize $r_i^{(0)} \leftarrow 1/n$ for $1 \leq i \leq n$
\For {$k = 1$ to $l$}
\For {$i = 1$ to $n$}
\State $r_i^{(k)} \leftarrow \sum_{j=1}^{n} N_{ji} r_j^{(k-1)}$ \Comment{update PageRank values}
\EndFor
\EndFor
\end{algorithmic}
\label{alg:pseudo_iterative}
\end{algorithm}

The iterative method to compute PageRank is oblivious when the input graph is represented using the adjacency matrix. As discussed earlier, adjacency matrix representation is not well suited for sparse graphs. To support a more efficient graph representation, an equivalent algorithm to compute PageRank for the case when the input graph is in the adjacency list format is provided  below. This algorithm is non-oblivious since the number of times the \textit{for} loop in step 5 is executed depends on the degree of all the vertices. The memory accesses on the entries of the matrix $N$ in step 6 of the algorithm are non-oblivious as well.

\begin{algorithm}
\caption{PageRank algorithm (Adjacency list representation)}
\begin{algorithmic}[1]
\Require{Graph $G(V,E)$, $l$ (number of iterations), $s$ (redistribution parameter)}
\Ensure{The PageRank values for all the vertices in the graph $G$ }	
\State Compute the matrix $N = [N_{ij}]$ as specified in Definition \ref{update_matrix}.
\State Initialize $r_i^{(0)} \leftarrow 1/n$ for $1 \leq i \leq n$ (all other $r^{(k)}$ vectors are initialized as zero vectors for $k \geq 1$)
\For {$k = 1$ to $l$}
\For {$i = 1$ to $n$}
\For {$j \in Adj(i)$ } \Comment{loop over neighbors of $i$}
\State $r_j^{(k)} \leftarrow r_j^{(k)}+ N_{ij} r_i^{(k-1)}$ 
\EndFor
\EndFor
\EndFor
\end{algorithmic}
\label{alg:pseudo_pagerank_adjlist}
\end{algorithm}

Further we modify the above pseudocode  to obtain a data-oblivious PageRank algorithm (Protocol \ref{alg:oblivious_pagerank}) that assumes the edgelist representation for the input graph. Major techniques employed to convert the pseudocode into its oblivious counterpart include loop coalescing and the ORAM primitive. Next we present formal proof of correctness and obliviousness for the designed protocol.

\begin{protocol}
\caption{Oblivious  PageRank algorithm (Edgelist representation)}
\begin{algorithmic}[1]\label{oblivious_pagerank}
\Require{Graph $(\textcolor{Violet}{\mathcal{E}}, \textcolor{Violet}{Idx})$, $\textcolor{ForestGreen}{l}$ (number of iterations), $\textcolor{red}{s}$ (redistribution parameter)}
\Ensure{The PageRank values for all the vertices in the graph $G$ }	
\For { $\textcolor{ForestGreen}{i} = \textcolor{ForestGreen}{1}$ to $\textcolor{ForestGreen}{n}$ }
\For {$\textcolor{ForestGreen}{j} = \textcolor{ForestGreen}{1}$ to $\textcolor{ForestGreen}{n}$}
\If { $\textcolor{Violet}{Idx}[\textcolor{ForestGreen}{i}]$ == $\textcolor{Violet}{Idx}[\textcolor{ForestGreen}{i}+\textcolor{ForestGreen}{1}]$} 
\State $\textcolor{Violet}{\overset{\sim}{N}}[\textcolor{ForestGreen}{i}][\textcolor{ForestGreen}{j}] \leftarrow (\textcolor{ForestGreen}{1}-\textcolor{red}{s})/\textcolor{ForestGreen}{n} $
\Else
\State $\textcolor{Violet}{\overset{\sim}{N}}[\textcolor{ForestGreen}{i}][\textcolor{ForestGreen}{j}] \leftarrow (\textcolor{red}{s})/(\textcolor{Violet}{Idx}[\textcolor{ForestGreen}{i}+\textcolor{ForestGreen}{1}] - \textcolor{Violet}{Idx}[\textcolor{ForestGreen}{i}])+(\textcolor{ForestGreen}{1}-\textcolor{red}{s})/\textcolor{ForestGreen}{n}$
\EndIf
\EndFor
\EndFor
\State  $\textcolor{Violet}{\overset{\sim}{r}^{(0)}}[\textcolor{ForestGreen}{i}] \leftarrow \textcolor{ForestGreen}{1}/\textcolor{ForestGreen}{n}$ for $\textcolor{ForestGreen}{1} \leq \textcolor{ForestGreen}{i} \leq \textcolor{ForestGreen}{n}$ 
\State  $\textcolor{Violet}{\overset{\sim}{r}^{(1)}}[\textcolor{ForestGreen}{i}] \leftarrow \textcolor{ForestGreen}{0}$ for $\textcolor{ForestGreen}{1} \leq \textcolor{ForestGreen}{i} \leq \textcolor{ForestGreen}{n}$ 
\State $\textcolor{ForestGreen}{j} \leftarrow \textcolor{ForestGreen}{1}$   
\For {$\textcolor{ForestGreen}{k} = \textcolor{ForestGreen}{1}$ to $\textcolor{ForestGreen}{l}$}
\State $\textcolor{red}{v} \leftarrow \textcolor{ForestGreen}{1}$ 
\For {$\textcolor{ForestGreen}{i} = \textcolor{ForestGreen}{1}$ to $\textcolor{ForestGreen}{m}$}   
\If {($\textcolor{ForestGreen}{i} == \textcolor{Violet}{Idx}[\textcolor{red}{v}+\textcolor{ForestGreen}{1}]$)} 
\State $\textcolor{red}{v} \leftarrow \textcolor{red}{v} + \textcolor{ForestGreen}{1}$ 
\EndIf
\State $\textcolor{Violet}{\overset{\sim}{r}^{(j)}}[\textcolor{Violet}{\mathcal{E}}[\textcolor{ForestGreen}{i}]] \leftarrow \textcolor{Violet}{\overset{\sim}{r}^{(j)}}[\textcolor{Violet}{\mathcal{E}}[\textcolor{ForestGreen}{i}]]+ \textcolor{Violet}{\textcolor{Violet}{\overset{\sim}{N}}}[\textcolor{red}{v}][\textcolor{Violet}{\mathcal{E}}[\textcolor{ForestGreen}{i}]] * \textcolor{Violet}{\overset{\sim}{r}^{(1-j)}}[\textcolor{red}{v}]$ 
\EndFor 
\State $\textcolor{ForestGreen}{j} \leftarrow (\textcolor{ForestGreen}{1} - \textcolor{ForestGreen}{j}) $ 
\State $\textcolor{Violet}{\overset{\sim}{r}^{(j)}}[\textcolor{ForestGreen}{i}] \leftarrow \textcolor{ForestGreen}{0}$ for $\textcolor{ForestGreen}{1} \leq \textcolor{ForestGreen}{i} \leq \textcolor{ForestGreen}{n}$
\EndFor
\end{algorithmic}
\label{alg:oblivious_pagerank}
\end{protocol}

\begin{theorem}
The proposed Protocol \ref{alg:oblivious_pagerank} correctly computes the PageRank centrality for all the vertices in the input graph $G(V,E)$.
\end{theorem}
\begin{proof}
In steps 1-6, the entries of the update matrix $N$ are computed in accordance with Definition \ref{update_matrix}. At all times we maintain two PageRank vectors, namely,  $\overset{\sim}{r}^{{\tiny(0)}}$ and  $\overset{\sim}{r} ^{(1)}$. One vector is used to update the pagerank in the current iteration, using the previous pagerank values stored in the other vector. For an odd valued $k$, the $k^{th}$ iteration of \textit{for} loop in step 10 updates the vector  $\overset{\sim}{r}^{(1)}$ in accordance with the PageRank update rule given in Definition \ref{update_rule}, assuming that  $\overset{\sim}{r}^{(0)}$ stores the PageRank values computed in the $(k-1)^{th}$ iteration i.e.  $\overset{\sim}{r}^{(0)} =  {r}^{(k-1)}$. In the $k^{th}$ iteration when $k$ is even, the roles of $\overset{\sim}{r}^{(1)}$ and $\overset{\sim}{r}^{(0)}$ will be reversed. The vector  $\overset{\sim}{r}^{(0)}$ will be updated in accordance with the PageRank update rule assuming that  $\overset{\sim}{r}^{(1)}$ stores the PageRank values computed in the $(k-1)^{th}$ iteration.
The two nested \textit{for} loops in steps 4-5 of the pseudocode for Pagerank, previously non-oblivious, are now coalesced into a single \textit{for} loop as given in step 12 of Protocol 2. 
It is easy to observe that the coalesced \textit{for} loop in protocol \ref{alg:oblivious_pagerank} and the nested \textit{for} loops in the pseudocode  will in total run for $m$ iterations, where $m$ denotes the number of edges in the graph.  
The coalesced \textit{for} loop visits an edge in each iteration and updates the pagerank of the node to which the edge points, using the update formula.   
Hence, the protocol correctly computes the PageRank values in accordance with  the Iterative PageRank method. 
\end{proof}

\begin{theorem}
The  Protocol 2  for computing the PageRank centrality is data-oblivious.
\end{theorem}
\begin{proof}
The algorithm is memory trace oblivious given that the arrays $N$, $\mathcal{E}$, $Idx$, $r^{(0)}$ and $r^{(1)}$ are stored in the ORAM. The \textit{if}-construct in step 3 can be made oblivious as discussed in preliminaries, Section \ref{prelim}. All the loops are over $n$, $m$ or $l$, all of which are public variables.	  
\end{proof}

\begin{theorem}
The  Protocol 2 takes $\mathcal{O}(( n^2 + lm)f(n^2) )$ primitive operations for computing the PageRank centrality.
\end{theorem}
\begin{proof}
In steps 1-8 of Protocol 2, we perform $\mathcal{O}(n^2)$ ORAM accesses over arrays of size $n$ and $n^2$. Hence the total cost will be $\mathcal{O}(n^2f(n^2))$ primitive operations. For each iteration of the \textit{for} loop on step 10 and step 12, we perform a constant number of block accesses on ORAM arrays of size $m$, $n$ and $n^2$ on step 15. Given that $m$ is the number of edges in the graph and it is upper bounded by $n^2$, the ORAM overhead for the constant number of accesses is $f(n^2)$. The theorem follows.  
\end{proof}

\begin{corollary}  .
The  Protocol 2  for computing the PageRank centrality uses $\mathcal{O}( (n^2 + ml) \log^2 (n) )$ primitive operations assuming  the ORAM accesses are performed using the Circuit ORAM construction.
\end{corollary}


\section{Realizing the Vision: Challenges and Future Directions}\label{challenges}

As a step towards realizing the proposed vision, we have designed data-oblivious algorithms for some of the most widely employed network analysis measures in the previous section. These can be converted into their equivalent secure multiparty protocols as discussed in Section \ref{prelim}. However, use of MPC as a practical solution to analyze DSSNs will require several open questions to be addressed. 

%
%
%
%

\subsection{Theoretical Challenges: Designing Efficient MPC Protocols}
In the 1980s and 1990s, many generic solutions were proposed for computing any arbitrary function $f$ securely. Though complete, these solutions cannot be put to practice attributing to the large overheads in the complexity of the MPC protocol over its non-secure variant. A recent methodology adopted for constructing efficient MPC protocols is to design data-oblivious algorithms. These oblivious algorithms can be converted into their MPC counterparts using ideal functionalities for all the primitive operations assumed in the oblivious algorithm. This approach has become popular since the introduction of the cryptographic primitive $\mathcal{F}_{ORAM}$ for performing efficient RAM model secure computation \cite{keller2015oblivious}. The same methodology is adopted in the current paper to design efficient MPC protocols for Google PageRank centrality and K-shell decomposition algorithm. Another lesser explored technique for designing MPC protocols is to make use of the fact that the output is disclosed in public at the end of an MPC protocol. Hence, one can design algorithms whose control flow and memory access pattern  depends on the output of the algorithm. This methodology is in contrast to the approach of designing data-oblivious algorithms, in which case the control flow and memory accesses of the program depend on nothing but the input length.  These protocols, though not data-oblivious, will still  give us secure MPC protocols. This methodology also brings out the non-equivalence of data-oblivious algorithms and secure computation protocols. A few works that have followed this particular approach and harnessed the leakage of output in an MPC protocol include the work by Brickell and Shmatikov \cite{brickell2005privacy}. Another approach to design secure MPC protocols could be to harness the \textit{security efficiency trade-off}. There may be scenarios where certain aspects of the input data are not sensitive and hence can be revealed in exchange for making the protocol more efficient. As an example, consider a communication network distributedly held by a set of individuals who wish to securely compute the K-shell decomposition on the distributed network. Hypothetically speaking, one could design MPC protocols that are not secure and disclose the degree distribution and average shortest path length of the input graph but are efficient compared to the secure MPC k-shell algorithm. Such a non-secure protocol can be used in practice, since it is previously known that the communication network is a small world and has a power law degree distribution.

\subsection{The Implementation Challenge}
As a field of cryptography, MPC has been studied for more than three decades. However, it is recent that MPC protocols have been  employed for real world applications. It was first used in 2008 by Danish farmers to conduct a double auction for determining the market price of Beetroot \cite{bogetoft2009secure}. Many implementations for MPC have been developed, which include Sepia, Sharemind, Fairplay, Tasty, VIFF and Oblivm. 
These implementations generally focus on compile-time optimizations, minimizing network communication and ensuring the robustness and correctness of implemented protocols. Till date MPC has been deployed in only those scenarios where the number of parties are small, generally two or three. However, to perform network analysis over a thousand or a million node network distributedly held by a large set of parties will require MPC protocols that are scalable and whose implementation is not prone to network congestion and network synchronization issues. Implementing such sizable MPC protocols for a large number of participants over the Internet has never been done before, and pose a great challenge for the MPC experts. There is, however, a feasible approach to handle data from a large set of participants. Here, we assume that the large number of participants distribute their data among a few external agents who then follow the MPC protocol on their behalf. Then, the MPC protocol must be designed by considering the semi-honest or malicious behavior of these agents rather than the participants. 


\subsection{The Challenge of Differential Privacy}
 
 Designing secure MPC protocols for functions whose input is distributedly held by a set of parties does not always ensure that the privacy of the involved individuals is maintained. At times the output of a secure protocol can by itself partially or completely reveal the private inputs of the parties. For example, consider a set of  10 individuals who wish to use an MPC protocol to output an unlabeled isomorphic random graph of the trust network knit on them. If the output graph has precisely one node with out-degree 3, then the corresponding individual (with out-degree 3) will be able to identify herself in the isomorphic shuffled network, and may even identify her neighbors using some auxiliary information. As observed in the previous example,  the output of an MPC protocol might reveal some or all of the private inputs of the participants. Hence one needs to also study which network functions are ``safe'' to compute. This question falls under the domain of \textit{differential privacy}, which aims at studying: how to maintain the utility of the output, while ensuring the privacy of the input. Although differential privacy has been extensively studied on relational data, it has been sparsely studied for networked data \cite{zhou2008brief}. There have also been instances where network data has been de-anonymized using some auxiliary information \cite{narayanan2009anonymizing}. Ensuring differential privacy of the networked input data is a much harder challenge than in the case of relational input data, given the adversary may have different types of information for de-anonymization, like edges, vertices, sub-graphs, quasi-identifiers, etc. Therefore, for ensuring the privacy of the individuals, there is also a need to study which network measures are differentially private to compute and which are not.

\subsection{The Deployment Challenge}
Implementing secure, robust and efficient MPC protocols for network analysis does not guarantee the participation of data holders on whom the network is knit. Generally, the complex security proofs of MPC  are accessible to only a small set of researchers in cryptography, hence there needs to be an incentive mechanisms in place, which must ensure that participants share their ``true'' private information with the developed secure application. In the MPC application deployed in \cite{bogetoft2009secure}, the Danish farmers had their own personal interest in computing the output of the MPC protocol (i.e. the market price at which they will sell their crops). However, such incentives may not be present if the MPC protocol is being employed for research purposes. This can be argued in the case, where the enmity network over a set of college students has to be analyzed using secure MPC protocols. In this case the participants will have no incentive to share their private data (list of enemies) with an MPC application whose sole purpose is data-analysis of DSSNs for research purposes.   Hence, to employ secure network analysis protocols in practice, a major challenge is to set up certain incentive mechanisms in place so that people can trust the developed application even when they do not benefit directly from the output of the MPC protocol. 


\section{Related Work}\label{related_work}

Classical graph algorithms in the context of MPC were first studied by Brickell and Shmatikov in 2006	\cite{brickell2005privacy}. The authors proposed secure two party protocols for computing single source and all pair shortest paths problem. Since then, many graph algorithms have been explored in the MPC setting, including BFS, DFS, Dijkstra, minimum spanning tree and classical flow algorithms \cite{aly2013securely,aly2014securely,blanton2013data}. Securely constructing an anonymized version of a network distributedly held by a set of parties has been previously studied for various security models \cite{frikken2006private,kukkala2016secure,tassa2013anonymization}. \\

There exist works that have aimed at designing secure MPC algorithms for a few network measures. These works have generally been motivated by some particular application scenario. To employ SNA for criminal investigation, Kerschbaum and Schaad \cite{kerschbaum2008privacy} propose secure multiparty protocols for closeness and betweenness centrality measures. They assume a different definition of betweenness centrality compared to the classical one specified by Brandes \cite{brandes2001faster}). Furthermore, they employ the adjacency matrix graph representation and hence their protocols are not very efficient for sparse graphs. Kukkala et al. \cite{kukkala2016privacy} propose secure multiparty protocols for degree distribution, closeness centrality, Pagerank and K-shell decomposition for the  adjacency matrix representation. Employing the ORAM primitive and the edgelist graph representation, the Pagerank and K-shell decomposition protocols designed in the current paper asymptotically outperforms the algorithms in the above work. Asharov et al. \cite{asharovsecure} study the a set of centrality measures over multilayer networks distributedly held by a set of individuals/organizations.  They propose information theoretic secure MPC protocols for distance based centrality measures.

\section{Conclusion}\label{conclusion}

In this paper, we highlight on the unexplored potential of studying distributed sensitive social networks. As our vision, we propose the use of the tools and techniques of multiparty computation to study the network properties of the distributedly held sensitive networks while ensuring  the privacy of involved individuals. The use of MPC overcomes all the drawbacks of the previously employed techniques (like surveys, interviews and sampling) for studying distributedly held networks. We present a list of theoretical, implementation, deployment and differential privacy challenges which need to be overcome to realize the proposed vision. Hence, a thorough investigation of all the network measures including community detection, degree distribution, betweenness centrality, eigenvector centrality, etc., in the MPC setting will require the confluence of different domains, including cryptography, security, network science and behavioral psychology.


\bibliographystyle{IEEEtran}

\bibliography{mybib}

\section*{Appendix}
\begin{subappendices}
\renewcommand{\thesection}{\Alph{section}}%


\section{Distributed Sensitive Social Networks}\label{appendix_dssn}

In this section we discuss on a set of DSSNs that have most often appeared in the research literature. For each network, we present its node-edge description, previous investigations on these networks, and a list of possible analysis that can be performed on them in future.\\   

\subsection*{Sexual Relationship Network}
The network consists of individuals modeled as nodes while the edges between them denotes the existence of a sexual relationship between the concerned individuals. This is a completely distributed network which has the highest quotient of sensitivity compared to other social networks. The works in the literature have focused on determining network characteristics such as assortativity and show that positive and negative assortativity induce different spreading patterns in the network \cite{rocha2010information}. It has also been observed that some of these networks deviate from the traditional model of preferential attachment that gives rise to the power law degree distribution \cite{rocha2010information}. The data for all the studies conducted so far in the literature have been gathered through face-to-face interviews, random sampling, etc that involve a trusted third party \cite{hivurl}.  


\subsection*{Romantic Relationship Network}
Another distributed sensitive social network that models relationships similar to the sexual relationship network is the romantic relationship networks. It has a high quotient of sensitivity at 0.85 and is currently available only in the fully distributed form. The nodes continue to be individuals while the presence of an edge will denote that the corresponding individuals have dated each other in the past or are currently dating. Studies on romantic relationships among school children have shown that nodes that represent romantic partners tend to have similar structural characteristics such as centrality indexes \cite{kreager2016friends}. However, such a study is not easy to perform when it comes to workplace romance \cite{forbesoffice}. Due to the presence of negative perceptions, workers usually tend to hide their workplace relationships. However, these relationships are known to influence the workplace environment and hence are needed to be investigated for a better performance of the organization. Given the highly debated pros and cons of workplace romance, the study of this DSSN can throw some light into the dynamics of these interactions on the work culture. 


\subsection*{Informal Networks - Trust and Enmity Networks}
Enmity network and trust network are often referred to as \textit{informal networks}, given that they model the informal interactions between the employees, which is usually inconsistent with the interactions induced by the formal hierarchy of the organization. These networks model individuals as nodes and the presence of an edge from node $A$ to node $B$ will denote the feeling of hatred in an enmity network and trust in a trust network. These networks also fall under the category of completely distributed social networks. Unlike the previously considered networks, the feelings of trust and hatred are not necessarily mutual. Hence the underlying network is best modeled as a directed graph. The feeling of animosity between colleagues does not foster a healthy working environment, hence it is of great interest for the organization to monitor the network of hatred on their employees. A study of the dynamics of hate relations would also help form more effective teams and take constructive measures to foster a better environment. There has not been any work, to the best of our knowledge, that investigates the enmity network. On the other hand, the importance of the study of the trust network of employees in an organization has appeared in the literature \cite{krackhardt200116}. There are instances where the progress of a team has been slow when the leader occupies a  weak position in the trust network of the organizations \cite{henry2001creative}. There are third party agencies, like \textit{Keyhubs}, who offer  services to perform workplace social analytics and allow harnessing the potential of informal networks \cite{keyhubs}. They identify influential employees in the network, unravel the team dynamics and help bring in a transformation by extracting structural metrics of the informal network. They collect the data regarding informal networks through surveys, and hence only act as a data processor on behalf of their clients. Such a method of analysis boils down to having a trusted third party model, wherein Keyhubs learns the entire data gathered for analysis. \\

\subsection*{Supply Chain Network}
Supply chain networks is another example of a fully distributed DSSN, where the nodes represent organizations that are a part of a buyer-seller system and a directed edge from $A$ to $B$ denotes that $A$ supplies goods/raw materials to $B$. Such a network is considered sensitive since organizations do not wish to disclose their buyer-seller interactions in the fear of compromising competitive advantage \cite{fridgen2015supply}. At the same time, organizations are interested in determining the  importance of the structural position they occupy in the supply chain since it has been well established that the position of an organization in the supply chain can greatly determine its influence over the flow of goods/raw materials in the entire network. There are works in the literature that look at assessing the risk associated with an organization based on its structural position \cite{kerschbaum2011secure,fridgen2015supply}. This is done by using secure multiparty protocols designed to compute the betweenness centrality of the nodes in the network. \\

\subsection*{Financial Transaction Network}
Financial transaction network is an example of a partially distributed DSSN. As the name suggests, this network models the financial transactions between the considered entities. The nodes in the network represent individuals with a bank account. An edge from $A$ to $B$ will exist in this network if $A$ has initiated more than a threshold number of transactions to $B$. One can also model the network as a weighted network, where the edge weights denote the frequency of transactions between the corresponding nodes. The accounts of the individuals who are a part of the network may belong to different banks. Each bank has the information of the transactions initiated by it's own customers. Hence, each bank is aware of the partial data (sub-graph) of the complete transaction network.  A preliminary analysis of the topology of the financial transactions network in Austria 
has been performed previously \cite{kyriakopoulos2009network}. \\

\section{Secure Computation and Ideal Functionalities}\label{intro_mpc}

\subsection{A Introduction to MPC}
Secure multiparty computation (MPC) is a sub-field of cryptography that deals with the design of  protocols/algorithms that allows a set of $n$ individuals/parties $P_1, P_2, \ldots, P_n$ with private inputs $x_1, x_2, \ldots, x_n$ to compute a public function $f(x_1, x_2, \ldots, x_n)$. The protocol must be  such that the parties don't learn any information from the run of the protocol, other than what can be gathered from the output of the protocol and their respective inputs. The design of the protocol must also take the behavior of the parties into consideration. As discussed previously, the behavior of the parties could be termed either as honest or as corrupt. We define the notion of an \textit{adversary} in order to model any attack on the protocol launched by the corrupt parties. The adversary is assumed to control the set of corrupt parties in the protocol. Instead of modeling corruption of individual corrupt parties,  the security of an MPC protocol can be be defined by modeling the adversary 
accordingly. For example, the adversary can be assumed to be either \textit{passive} or \textit{active}. A \textit{passive} adversary can only read the messages sent/received by the corrupt parties and cannot influence their behavior, whereas, an \textit{active} adversary can influence the actions of the corrupt parties during the run of the protocol. The security of an MPC protocol is defined by comparing the ideal setting and the real simulation, details of which are described below. \\

The ideal setting refers to the scenario of computing a functions $f$ where we assume the existence of a trusted third party $TTP$, to whom all the parties $P_1, P_2, \ldots, P_n$ send their private inputs $x_1, x_2, \ldots, x_n$ respectively. The $TTP$ computes the function $f(x_1,x_2,\ldots,x_n)$ and sends the desired output to all the parties. This protocol for  computing $f$ in the presence of a TTP is denoted by $F_{\pi}$. The protocol $F_{\pi}$ is termed as the \textit{ideal functionality} for the protocol $\pi$, where $\pi$ denotes the MPC protocol designed to achieve the secure computation of $f$.

In the real setting we run the protocol $\pi$, as a part of which the set of parties $P_1,P_2,\ldots,P_n$ perform computation, exchange messages, make random choices, etc.
The sequence of messages that the party sent/received and the random choices that the party made during the run of the protocol is defined as the \textit{view} of a party. Intuitively, an MPC protocol $\pi$ \textit{securely} computes the function $f$ if the adversary gains no extra advantage when the function $f$ is computed using the protocol $\pi$ over the its ideal functionality $F_{\pi}$. In the case of a \textit{semi-honest} or a \textit{passive} adversary, we can formally define security as follows: \\

\begin{definition}{(Passively secure MPC protocol)}\label{security_definition} \\
An MPC protocol $\pi$ securely computes a function $f$ if there exist a \textit{simulator} $S$ i.e. an efficient probabilistic algorithm which generates the view of the adversary in the real world given just the input and output of the corrupt parties i.e.
\begin{align*}
\{view_i\}_{P_i \in \mathcal{C}} \equiv S(\{x_i,y_i\}_{P_i \in \mathcal{C}})
\end{align*}
where $\mathcal{C}$ represents the set of corrupt parties and $view_i,x_i,y_i$ represents the view, input and output of party $P_i$ respectively.
\end{definition}  
\mbox{ }  \\
On similar lines we can define security of MPC protocols for the case of active adversaries. A more thorough discussion on the security of MPC protocol can be found in \cite{mpcbook}. \\

The efficiency of an MPC protocol is generally measured on the basis of three criteria, namely, computation, communication and round complexity. The \textit{computation complexity} accounts for the total computation that all the parties need to perform before and during the execution of the protocol. The \textit{communication complexity} is a measure of the net amount of information exchanged between the parties.
 A secure multiparty protocol can be broken down into rounds, such that those instructions which can be executed in parallel and are independent of each other are clubbed as a single round of execution.  
The total number of rounds employed during the run of the protocol accounts for the \textit{round complexity}. 

\subsection{Universal Composability (UC) Framework}
For all the secure multiparty protocols designed in the paper, few functionalities/sub-protocols will be repeatedly used in and across protocols. Hence, we will abstract out these sub-protocols as a set of ideal functionalities which will be available as primitive operations. This abstraction is feasible given the UC theorem \cite{canetti2001universally}, which states that a protocol $\pi_1$ is secure when composed with a protocol $\pi_2$ if the protocol $\pi_1$ can be securely implemented when composed with the ideal functionality $F_{\pi_2}$. The UC theorems allows us to design protocols using ideal functionalities for protocols that are already implemented in the literature. Further in this section we will present a set of ideal functionalities which will be employed for designing secure MPC protocols for network measures.

\subsection{Arithmetic Black Box $\mathcal{F}_{ABB}$ }
The arithmetic black box $\mathcal{F}_{ABB}$ is one of the most primitive ideal functionalities employed while designing complex MPC protocols. The ideal functionality is characterized by a ring $\mathbb{Z}_M$ and it provides the following four operations:

\begin{itemize}
\item \textbf{Store}: A party $P$ can store an element $a \in \mathbb{Z}_M$ in the black box by using the following command: 
\begin{align*}
[a] \leftarrow_P a 
\end{align*}
The handle $[a]$ is sent to all the parties and the stored value $a$ can be manipulated using the handle $[a]$.
\item \textbf{Addition}: All the parties can securely add ($+_M$) the values stored at handle $a$ and $b$ and store the result in the location with the handle $c$, using the following command:
\begin{align*}
[c] \leftarrow [a] + [b] 
\end{align*}
\item \textbf{Multiplication}: All the parties can securely multiply ($*_M$) the values stored at handle $a$ and $b$ and store the result in the location with the handle $c$ using the following command:
\begin{align*}
[c] \leftarrow [a] * [b] 
\end{align*}
\item \textbf{Release}: A value $a$ stored at handle $[a]$ can be released in public by all the parties using the following command:
\begin{align*}
a \leftarrow [a] 
\end{align*}
\end{itemize}

The arithmetic back box $\mathcal{F}_{ABB}$ has been realized for various security models using cryptographic primitives like secret sharing \cite{rabin1989verifiable} and homomorphic encryption \cite{damgaard2003universally}.  We can further extend this ideal functionality with a few more commonly used and previously implemented operations. We append the basic ideal functionality $\mathcal{F}_{ABB}$ with the following operations:
\begin{itemize}
\item \textbf{Comparison}: All the parties can securely compare the values stored at handle $a$ and $b$ and store $1$ or $0$ at handle $c$ depending on whether $a$ is greater than $b$ or not using the following command:
\begin{align*}
[c] \leftarrow [a] \overset{?}{>}  [b] 
\end{align*}
\item \textbf{Equality}: All the parties can securely check for equality of the values stored at handle $a$ and $b$ and store $1$ or $0$ at handle $c$ depending on whether $a$ equals $b$ or not using the following command:
\begin{align*}
[c] \leftarrow [a] \overset{?}{=}  [b] 
\end{align*}
\end{itemize}

The extension of the basic arithmetic black box to allow the above mentioned operations has been well explored in the MPC literature \cite{damgaard2006unconditionally}. For the sake of brevity, we will refer to the extended arithmetic black box as the arithmetic black box $\mathcal{F}_{ABB}$ itself.\\

The exact computation, communication and round complexity associated with these operation vary depending on the underlying implementation of these operations. Furthermore, optimization of the complexity of the operations of the ideal functionality $\mathcal{F}_{ABB}$ has been a constant area of research \cite{lipmaa2013secure}. To make the MPC protocol designed in the paper independent of the implementation of the above mentioned operations we will evaluate the complexity of each designed protocol by the number of $\mathcal{F}_{ABB}$ operations invoked during the run of the protocol. This will also provide us with a simpler metric to compare the complexity of the designed MPC protocols with their equivalent non-secure variants. 

\subsection{Oblivious RAM for Secure Computation $\mathcal{F}_{ORAM}$}
The oblivious RAM (ORAM) primitive was first introduced by Goldreich and Ostrovsky in their seminal paper published in 1996 \cite{goldreich1996software}. Consider the scenario where a client has stored her data on a remote server in an $n$ length array $A$. Further, the client wishes to query the block $A[i]$ of the array without leaking the index $i$ to the server. The ORAM primitive allows the client to access the block $A[i]$ in time $poly.log(n)$ while ensuring that the memory access pattern of the array does not leak the index $i$ to the server. Since 1996, the ORAM primitive has been employed in various domains including data outsourcing \cite{stefanov2011towards}, secure processors \cite{maas2013phantom} and secure computation \cite{gentry2013optimizing}. In the secure computation scenario, the ORAM primitive has been employed to implement the ideal functionality $\mathcal{F}_{ORAM}$, which provides the following operations:
\begin{itemize}
\item \textbf{Initialization}: All the parties can initialize a new $n$ length array in the $\mathcal{F}_{ORAM}$ using the following command:
\begin{align*}
[A] \leftarrow oram\_initialize(n)
\end{align*}
The handle $[A]$ to the newly initialized array is sent to all the parties. 
\item \textbf{Read}:  All the parties can read the $i^{th}$ block of the array $A$ securely given $[i]$ using the following command:
\begin{align*}
[x] \leftarrow oram\_read([A], [i])
\end{align*}
The handle $[x]$ contains the content $A[i]$.
\item \textbf{Write}: All the parties can write the content stored at $[x]$ to the $i^{th}$ block of the array $A$ securely given $[i]$ using the following command:
\begin{align*}
oram\_write([A], [i], [x])
\end{align*}
\end{itemize}
Numerous implementations of $\mathcal{F}_{ORAM}$ exist where the $oram\_read$ and $oram\_write$ operations mentioned above can be implemented using $poly.log (n)$ operations of the $\mathcal{F}_{ABB}$ \cite{wang2014scoram}. Circuit ORAM designed by Shi et al. \cite{wang2015circuit} is till date the best known implementation of $\mathcal{F}_{ORAM}$. Circuit ORAM allows one to perform the read and write operations for arbitrary block sizes using only $\mathcal{O}((\log n)^2)$ $\mathcal{F}_{ABB}$ operations. The ORAM initialization cost for the Circuit ORAM is $\mathcal{O}(n(\log n))$. For all the protocols designed in the paper the ORAM initialization cost for all the ORAM array used will be amortized over the set of read/write operations performed on the ORAM arrays.

\subsection{Secure Computation over Rational Numbers $\overset{\sim}{\mathcal{F}}_{ABB}$}

In numerous network analysis algorithms, computation is performed over fractional values. Since the arithmetic black box $\mathcal{F}_{ABB}$ has so far only allowed integral operations on $\mathbb{Z}_M$, there is a need to extend the  $\mathcal{F}_{ABB}$ to perform computation over rational numbers. Catrina and Saxena \cite{catrina2010secure} introduced the ideal functionality $\overset{\sim}{\mathcal{F}}_{ABB}$ which provides basic arithmetic operations including addition, multiplication and division over rational numbers in the fixed-point representation. A fixed point number has two parts, the integer and the fractional part separated by the radix point. The ideal functionality $\overset{\sim}{\mathcal{F}}_{ABB}$ is characterized by specifying the number of decimal bits required in the integral and the fractional part. A rational number will be distinguished from an integer by using the symbol tilda ($\sim$) over the handle of the variable. The primitive operations provided by the ideal functionality $\overset{\sim}{\mathcal{F}
}_{ABB}$ are specified below:

\begin{center}
\begin{tabular}{l l l  }
1. & Store &\hspace*{2cm}  $[\overset{\sim}{a}] \leftarrow_P \overset{\sim}{a} $  \\ 
2. & Addition &\hspace*{2cm}  $[\overset{\sim}{c}] \leftarrow [\overset{\sim}{a}] + [\overset{\sim}{b}]$  \\  
3. & Multiplication &\hspace*{2cm} $[\overset{\sim}{c}] \leftarrow [\overset{\sim}{a}] * [\overset{\sim}{b}]$ \\
4. & Division &\hspace*{2cm} $[\overset{\sim}{c}] \leftarrow [\overset{\sim}{a}] / [\overset{\sim}{b}]$ \\
5. & Comparison &\hspace*{2cm} $[\overset{\sim}{c}] \leftarrow [\overset{\sim}{a}] \overset{?}{>}  [\overset{\sim}{b}]$ \\
6. & Equality &\hspace*{2cm} $[\overset{\sim}{c}] \leftarrow [\overset{\sim}{a}] \overset{?}{=}  [\overset{\sim}{b}]$ \\
7. & Release &\hspace*{2cm} 	$\overset{\sim}{a} \leftarrow [\overset{\sim}{a}]$
\end{tabular}
\end{center}

\end{subappendices}

\end{document}